\documentclass[letterpaper,10pt,conference]{ieeeconf}  

\IEEEoverridecommandlockouts                              
\usepackage{graphics} 
\usepackage{times} 
\usepackage{amsmath} 
\usepackage{amssymb}  
\usepackage{amsfonts}
\usepackage[acronym]{glossaries}
\usepackage{quantikz}
\usepackage[hidelinks]{hyperref}
\usepackage[raggedright]{subfigure}

\usepackage[capitalize,noabbrev]{cleveref}
\Crefname{equation}{}{}
\Crefname{figure}{Fig.}{Figs.}
\Crefname{figure}{Tab.}{Tabs.}
\Crefname{section}{Sec.}{Secs.}
\Crefname{appendix}{App.}{Apps.}

\newtheorem{theorem}{Theorem}
\newtheorem{theorem*}{Theorem}

\newif\ifshort

\begin{document}

\newacronym{qec}{QEC}{quantum error correction}
\newacronym{ftqc}{FTQC}{fault-tolerant quantum computing}
\newacronym{cqsk}{C-QSK}{Clifford quantum simulation kernel}
\newacronym{lcs}{LCS}{\emph{logical Clifford synthesis}}
\newacronym{sas}{SAS}{\emph{solve-and-stitch}}
\newacronym{pbs}{PBS}{\emph{piecewise-best selection}}

\title{\LARGE \bf
    Optimized Compilation of Logical Clifford Circuits 
}

\author{Alexander Popov$^{*,1,2}$, Nico Meyer$^{1,3}$, Daniel D. Scherer$^{1}$ and Guido Dietl$^{2}$
    \thanks{The research was supported by the Bavarian Ministry of Economic Affairs, Regional Development and Energy with funds from the Hightech Agenda Bayern via the project BayQS and in parts by the German Federal Ministry of Research, Technology and Space, funding program Quantum Systems, via the project Q-GeneSys, grant number 13N17389.}
    \thanks{$^1$Fraunhofer IIS, Fraunhofer Institute for Integrated Circuits IIS, N\"urnberg, Germany; $^2$University of W\"urzburg, Professorship of Satellite Communication and Radar Systems, W\"urzburg, Germany; $^3$Pattern Recognition Lab, Friedrich-Alexander-Universit\"at Erlangen-N\"urnberg, Erlangen, Germany. $^*$Email: \texttt{alexander.popov.res@proton.me}}
}

\maketitle

\begin{abstract}
Fault-tolerant quantum computing hinges on efficient logical compilation, in particular, translating high-level circuits into code-compatible implementations. Gate-by-gate compilation often yields deep circuits, requiring significant overhead to ensure fault-tolerance. As an alternative, we investigate the compilation of primitives from quantum simulation as single blocks. We focus our study on the $[[n,n\mathrm{-}2,2]]$ code family, which allows for the exhaustive comparison of potential compilation primitives on small circuit instances. Based upon that, we then introduce a methodology that lifts these primitives into size-invariant, depth-efficient compilation strategies. This recovers known methods for circuits with moderate Hadamard counts and yields improved realizations for sparse and dense placements. Simulations show significant error-rate reductions in the compiled circuits. We envision the approach as a core component of peephole-based compilers. Its flexibility and low hand-crafting burden potentially enable extension to other circuit structures and code families.
\end{abstract}

\glsresetall
\bstctlcite{IEEEexample:BSTcontrol}  

\section{\label{sec:intro}Introduction}

\Gls{qec} is essential for scaling quantum computers and remains a primary bottleneck~\cite{shor1995scheme,google2025quantum}. A key challenge is efficient logical compilation: mapping circuits to a chosen code with low undetected error while respecting architectural constraints and minimizing depth and space–time volume~\cite{gottesman2009introduction,litinski2019game}. For leading approaches such as the surface code, gate-by-gate compilation inflates resource requirements because logical gates require repeated stabilizer measurement rounds and code-deformation primitives (e.g., lattice surgery), introducing routing, ancillas, and synchronization barriers~\cite{horsman2012surface,gidney2019efficient}. Global optimization that exploits commutation and shared parity measurements is therefore crucial to remove unnecessary cycles and active patches, reduce logical depth and space–time volume, and align with surrounding subroutines such as distillation and syndrome extraction.

Building on the stabilizer framework and early fault-tolerant gate realizations~\cite{gottesmann1998fault,calderbank1996good}, and on binary symplectic representations of Clifford operations with encoded Paulis~\cite{dehaene2003clifford}, \emph{Rengaswamy et al.}~\cite{rengaswamy2018synthesis} introduced a symplectic-geometry synthesis that enumerates compatible logical realizations for arbitrary stabilizer codes. While this \gls{lcs} approach exposes the full design space, its search space grows combinatorially with circuit width and gate count, making it costly to isolate low-depth, low-volume realizations at scale. More recently, \emph{Chen and Rengaswamy}~\cite{chen2024tailoring,chen2025tailoring} proposed a complementary, algorithm-tailored \gls{sas} formulation to hand-craft depth-efficient logical implementations of \glspl{cqsk}~\cite{li2022paulihedral}\textemdash{}an important primitive in quantum simulation algorithms\textemdash{}on the \([[n,n\mathrm{-}2,2]]\) code family~\cite{gottesmann1998fault}. This reduces compilation complexity via scalable templates but requires substantial manual design and, as we show, can be suboptimal in some regimes.

As in Refs.~\cite{rengaswamy2018synthesis,chen2024tailoring,chen2025tailoring,chen2025fault}, we focus on \gls{cqsk}, which arise as well-defined subroutines inside larger quantum circuits and can be optimized independently of surrounding non-Clifford gates. Our constructions thus serve as logical peephole optimizations~\cite{liu2021relaxed,rietsch2024unitary}: we identify and replace Clifford sub-circuits with depth-optimized encoded realizations on the chosen code family. Concretely, we mine small instances with \gls{lcs}~\cite{rengaswamy2018synthesis} to discover low-depth compilation primitives, then lift these patterns into strategies that scale to arbitrary system sizes while respecting code constraints. This narrows the search to structurally meaningful candidates, reduces hand-crafting, and improves upon \gls{sas}~\cite{chen2024tailoring,chen2025tailoring} in edge regimes.

\medskip
\noindent\textbf{Main Contributions of this Work.} We claim and
summarize the main contributions of our paper
as follows:
\renewcommand{\labelenumi}{\Roman{enumi}.}
\begin{enumerate}
    \item We systematically mine the design space of small \gls{cqsk} instances using \gls{lcs}, uncovering depth-favorable compilation primitives and distilling them into closed-form sequencing rules for the \([[n,n\mathrm{-}2,2]]\) family.
    \item We generalize these components into scalable compilation strategies that synthesize circuits of arbitrary size with reduced manual effort, enabling potential extensions to other code families and circuit structures.
    \item We introduce proof techniques that verify correctness by checking the required logical Pauli constraints and stabilizer preservation.
    \item We empirically validate the derived strategies under noise simulation, demonstrating reduced depth and higher success rates; the largest gains over \gls{sas} appear in edge regimes of \gls{cqsk}.
\end{enumerate}
Overall, we provide a mathematical framework for the efficient development of depth-optimized compilation strategies, tailored for peephole optimization of Clifford sub-circuits.

\medskip
\noindent
The remainder of this paper is structured as follows: In \cref{sec:preliminaries}, we establish preliminaries on stabilizers, Pauli constraints, and the \gls{cqsk}. We furthermore position our work in the context of quantum compilation and summarize necessary insights from prior work, in particular regarding the \gls{lcs} and \gls{sas} approaches. In \cref{sec:compilation_primitives}, we use the aforementioned \gls{lcs} algorithm to identify promising compilation primitives on tractable system sizes. The identified structures are generalized in \cref{sec:compilation_strategy} to arbitrary system sizes, and we furthermore prove correctness. In particular, we re-discover the \gls{sas} approach under certain structural assumptions on the \gls{cqsk}, while we identify improved strategies for relaxed restrictions. In \cref{sec:evaluation}, we perform an empirical evaluation of the discovered compilation strategies and quantify improvements over \gls{sas}. Finally, in \cref{sec:discussion}, we position our results in the greater context of \gls{ftqc}, discuss limitations, and identify potential for future extensions.


\section{\label{sec:preliminaries}Preliminaries and Prior Work}

We focus on Clifford sub-circuits within larger algorithms. Clifford circuits map the Pauli group $\mathcal{P}_n$ onto itself. As every quantum state decomposes into Pauli operators, a Clifford circuit is fully characterized by its action on the Pauli basis. It suffices to track the generators $X_{1:k}$ and $Z_{1:k}$; the transformations of $Y_{1:k}$ follow from $Y_i=X_iZ_i$ for $i\in\{1, \dots, k\}$. We call these transformation rules the \emph{Pauli constraints}. A logical implementation on any $[[n,k,d]]$ stabilizer code, where $n$ is the number of physical qubits, $k$ the number of encoded logical qubits, and $d$ the code distance, must satisfy two key requirements:

\subsubsection{Logical Pauli Constraints}  
The circuit must mimic the mapping relations of the Clifford circuit for the \( k \) logical Pauli generators \( \overline{X}_i \) and \( \overline{Z}_i \). These logical constraints are given by:
\begin{align}
    \begin{aligned}
    \overline{X}_i &\mapsto C_i^X(\overline{X}_{1:k},\,\overline{Y}_{1:k},\,\overline{Z}_{1:k}),\\[2mm]
    \overline{Z}_i &\mapsto C_i^Z(\overline{X}_{1:k},\,\overline{Y}_{1:k},\,\overline{Z}_{1:k}).
    \end{aligned}
\end{align}
Here, $\overline{X}_{1:k}$ abbreviates $\overline{X}_1\ldots \overline{X}_k$, and $C_i^X$ and $C_i^Z$ describe how each Pauli operator transforms under the C-QSK circuit.
As the logical $\overline{X}_i$ and $\overline{Z}_i$ operations on the $[[n,n\mathrm{-}2,2]]$ code family are $X_1X_{i+1}$ and $Z_{i+1}Z_{n}$~\cite{gottesmann1998fault,chao2018fault}, respectively, we define the embedding $\mathcal{E}:[k]\to\{2,\ldots,n-1\}$ by $\mathcal{E}(i)=i+1$ to simplify notation. 

\subsubsection{Stabilizer Preservation}  
To enable reliable error correction and detection, the logical Clifford circuit must preserve the stabilizer group  $\mathcal{S}$  of the code, meaning it maps $\mathcal{S}$ back onto itself under conjugation. This preservation of the stabilizer group ensures that the syndrome measurements remain valid throughout the computation.

\smallskip
\noindent
Because globally optimal compilation is impractical, one typically uses local peephole-based optimizations~\cite{liu2021relaxed,rietsch2024unitary}. These local blocks appear in quantum simulation kernels—repeating, frequently executed subroutines of quantum simulation algorithms~\cite{li2022paulihedral}. We focus on \glsentrylongpl{cqsk} (C-QSKs), where $R_z(\theta)$ is restricted to $P = R_z(\frac{\pi}{2})$, as shown in \cref{fig:CQSK}. For a fixed system size, these circuits differ only in the number and locations of Hadamard gates; let $I_h \subseteq \left\{ 1,\ldots,k\right\}$ denote the set of qubit indices with Hadamards, where $h$ is its cardinality.

To date, logical compilation is largely gate-by-gate, with few exceptions for specialized circuit classes like IQP circuits~\cite{paletta2024robust,hangleiter2025fault}. While broader \gls{qec} tasks have seen rapid progress\textemdash{}e.g., ML-based encodings~\cite{meyer2025learning,meyer2025variational,meyer2026learning}, gadgets~\cite{meyer2026logical}, and decoders~\cite{bausch2024learning}\textemdash{}efficient, code-specific compilation remains underexplored. Among the few targeted methods are \gls{lcs} and \gls{sas}, which we adopt and summarize below.

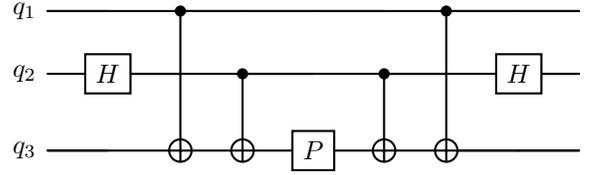
\begin{figure}[tbp]
   \centering
    \begin{minipage}{0.5\textwidth}
        \centering
        \scalebox{0.8}{
            \begin{quantikz}
                \lstick{$q_1$} & \qw  & \ctrl{2} & \qw & \qw & \qw & \ctrl{2} & \qw & \qw \\
                \lstick{$q_2$} & \gate{H} & \qw & \ctrl{1} & \qw & \ctrl{1} & \qw & \gate{H} & \qw \\
                \lstick{$q_3$} & \qw & \targ{} & \targ{} & \gate{P} & \targ{} & \targ{} & \qw & \qw
            \end{quantikz}
        }
    \end{minipage}%
    \caption{Clifford quantum simulation kernel (C-QSK) on three qubits with a Hadamard gate only on the second qubit, i.e.\ $I_h = \left\{ 2 \right\}$ and $\left| I_h \right| = 1$.}
    \label{fig:CQSK}
\end{figure}

\subsection{The Logical Clifford Synthesis (LCS) Algorithm}
The \gls{lcs} algorithm introduced in \cite{rengaswamy2018synthesis} generates all logical realizations of any Clifford circuit on any quantum code. 
To achieve that, first, the Pauli constraints of the underlying circuit are transformed into binary symplectic equations by means of the so-called $m$-qubit operator
\begin{align}
    D(a,b) = X^{a_1} Z^{b_1} \;\otimes\; \cdots \;\otimes\; X^{a_m} Z^{b_m},
\end{align}
with $a = [a_1,\ldots,a_m]$, $b = [b_1,\ldots,b_m] \in \mathbb{F}_2^{\,m}$, $m \in \mathbb{N}$. Every Pauli operator can be seen as such an $m$-qubit operator and therefore can be associated with the two binary vectors $a$ and $b$. Hence, the Pauli constraints can be interpreted as mappings from one set of binary vectors to another, which gives rise to the system of linear equations
\begin{align}
    [a_i', b_i'] = [a_i, b_i]\, F , \qquad i \in \{1,\ldots,k\},
\end{align}
where $k$ is the system size for the respective Clifford circuit.

\begin{figure*}[tbp]
  \centering
  \includegraphics[width=0.99\textwidth]{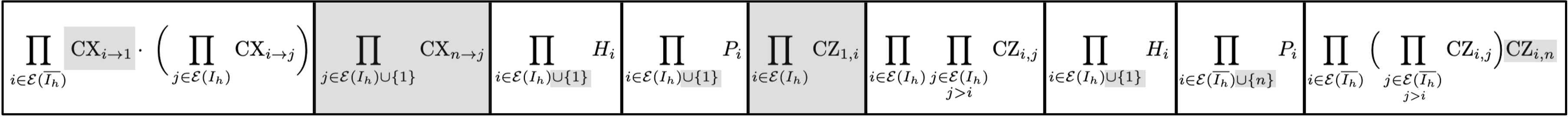}
  \caption{\label{fig:sas_rules}Full sequencing rule obtained by the $\mathcal{C}^{H}_{mid}$ compilation strategy, which is equivalent to the \gls{sas} from \cite{chen2024tailoring} (up to re-formulations with different notation). Operations highlighted in gray occur only for odd Hadamard configurations. The embedding function is given by $\mathcal{E}: [k] \rightarrow \{2,\ldots,n-1\}$ with $\mathcal{E}(i) = i+1$.}
\end{figure*}

The symplectic solution $F$ is then decomposed into a product of so-called elementary forms, which correspond to specific elementary quantum gates. This product therefore defines a Clifford quantum circuit for each symplectic solution obtained with the \gls{lcs} algorithm. It is important to mention that the decomposition into elementary forms is not unique and, therefore, variations in the realized circuits are possible. By design, the \gls{lcs} algorithm iterates all possible solutions of the system of equations. While this in principle allows to search for compilations that are optimal w.r.t.\ a given metric, the solution space grows super-exponentially in the stabilizer rank: for a general \([[n,k,d]]\) stabilizer code, a given logical Clifford circuit admits \(2^{r(r+1)/2}\) symplectic solutions with \(r = n-k\). To the best of our knowledge, no methods are known to filter for solutions with specific properties already at the solving stage. However, for the $[[n,n\mathrm{-}2,2]]$ code family considered in this work, there are only $2^3 = 8$ solutions, which allows for an exhaustive analysis of the solutions. 

\subsection{The Solve-and-Stitch (SAS) Algorithm}
The \gls{sas} algorithm described in \cite{chen2024tailoring,chen2025tailoring} is a specialized approach in order to obtain a logical realization of the \gls{cqsk} on the $[[ n, n\mathrm{-}2, 2]]$ code family. Intuitively, following the explicit form of the logical operators  $\overline{X}_{i} = X_{1} X_{i+1}$, $\overline{Z}_{i} = Z_{i+1} Z_{n}$, the Pauli constraints are written in a closed form giving rise to compilation strategies that generalize to arbitrary system sizes $k$.

The concrete procedure of the \gls{sas} algorithm is the following: given the logical-to-physical $X$- and $Z$-mappings, one first constructs, for each gate in the \gls{cqsk} separately, a circuit that realizes the corresponding mapping relation (the \emph{solve} step). These individual logical realizations are then \emph{stitched} together, typically by removing duplicate gates, so that the resulting circuit simultaneously fulfills all required logical-to-physical mapping relations. The combined circuit is subsequently checked for stabilizer preservation and, if necessary, adjusted. Applying this procedure yields a compilation strategy which we summarize using a slightly modified notation in \cref{fig:sas_rules}. This procedure, however, requires substantial manual effort for the construction of the individual mapping circuits, the stitching such that the combined circuit satisfies all constraints simultaneously, and the adjustments to ensure stabilizer preservation.


\section{\label{sec:compilation_primitives}Depth-Optimized Compilation Primitives for Small Systems}
Our objective is to design compilation routines for \gls{cqsk} circuits that produce physical realizations with shallow depth, and thus presumably reduced error susceptibility. Towards this goal, we first employ the above-discussed \gls{lcs} algorithm to systematically analyze the solution space of small circuit instances. As discussed above, for the $[[n,n\mathrm{-}2,2]]$ code family \gls{lcs} produces eight solutions, i.e.\ eight candidates for compilation primitives.

Apart from the number of qubits, the \glspl{cqsk} which we consider only differ by the number and position of the Hadamard gates. For this purpose, for small system sizes with $2 \le k \le 8$ logical qubits, we exhaustively enumerated all possible Hadamard placements. For medium-sized systems $10 \le k \le 20$, such an exhaustive sweep would have been computationally infeasible, since the number of placements grows combinatorially in $k$. In this case, we restrict the analysis to $10$ randomly sampled placements per Hadamard-count, which we found to still provide a representative picture of the circuit behavior.

\begin{figure}[tbp]
  \centering
  \includegraphics[width=0.8\linewidth]{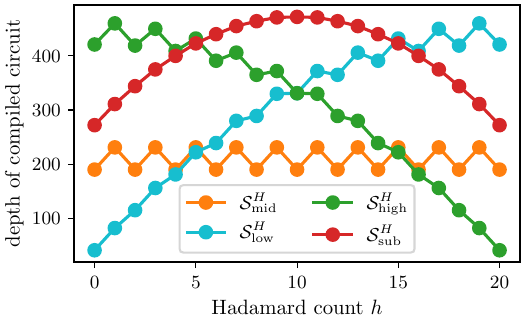}\caption{\label{fig:comparison_sol0_3}Comparison of depth for logical realizations of \glspl{cqsk} on varying system sizes as a function of the Hadamard count $h$. We plot the four position-invariant solutions of \gls{lcs} for \gls{cqsk} circuits on $k=20$ qubits. While solution $\mathcal{S}^{H}_{\text{mid}}$ produces particularly shallow circuits in the regime of about $5$ to $15$ Hadamard gates, $\mathcal{S}^{H}_{\text{low}}$ is clearly superior in case of sparse and $\mathcal{S}^{H}_{\text{high}}$ for dense Hadamard placement.}
\end{figure}

From this analysis, we found that the depth of the logical realization (mainly driven by the two-qubit count, as single-qubit gates are compiled in compact layers) strongly depends on which of the eight solutions from the \gls{lcs} algorithm is selected. Interestingly, for four of the solutions, the depth is invariant w.r.t.\ the actual Hadamard positioning and only depends on the count $h$. The other four solutions show a large spread of depth over different Hadamard placements. As the minimal depth among these realizations is in general significantly higher than for the well-behaved solutions, we focus only on these for our further analysis.

Extending our analysis to a system size of $k=20$ in \cref{fig:comparison_sol0_3}, we find that different solutions perform particularly well for different \gls{cqsk} structures. One solution produces minimal-depth circuits for moderate Hadamard counts of about $5$ to $15$, which we refer to as solution $\mathcal{S}^{H}_{\text{mid}}$. For circuits with $5$ or less Hadamard gates, depth is minimized by solution $\mathcal{S}^{H}_{\text{low}}$, for $15$ or more it is best to use $\mathcal{S}^{H}_{\text{high}}$. The remaining one of the four solutions $\mathcal{S}^{H}_{\text{sub}}$ is sub-optimal in all regions, so it can be eliminated from the set of solution candidates. Interestingly, as $k$ increases, the difference between the three remaining solutions becomes more pronounced in the respective Hadamard regions.

Above insights motivate us to use different \gls{lcs} solutions for different Hadamard counts, which we refer to in the following as the \glsentrylong{pbs} strategy. Evaluations show that  for $k > 10$, the crossover points are about
\begin{align}
&\text{Use } \mathcal{S}^{H}_{\text{low}} \text{ for } 0 < h < \tfrac{k}{4},\nonumber \\
&\text{use } \mathcal{S}^{H}_{\text{mid}} \text{ for } \tfrac{k}{4} \le h \le \tfrac{3k}{4}, \label{eq:selection_rules} \\
&\text{use } \mathcal{S}^{H}_{\text{high}} \text{ for } \tfrac{3k}{4} < h < k,\nonumber
\end{align}
with explicit analytic derivations planned for future work. Following the case distinction in \cref{eq:selection_rules}, the optimal compilation strategy can be selected purely based on the Hadamard count and size of the \gls{cqsk} of interest. However, currently, for each circuit we still have to use the \gls{lcs} algorithm to produce the desired logical realization. This is highly impractical, in particular for larger system sizes. We eliminate this inefficiency by generalizing the solution primitives to explicit closed-form compilation strategies in the next section.

\section{\label{sec:compilation_strategy}Generalized Compilation Strategy for Arbitrary System Sizes}

In this section, we generalize the three solution primitives $\mathcal{S}^{H}_{\text{low}}, \mathcal{S}^{H}_{\text{mid}}, \mathcal{S}^{H}_{\text{high}}$ identified in \cref{sec:compilation_primitives} to actual compilation strategies $\mathcal{C}^{H}_{\text{low}}, \mathcal{C}^{H}_{\text{mid}}, \mathcal{C}^{H}_{\text{high}}$. Therefore, we systematically examine the logical circuit structures for small system sizes $k$ and different Hadamard counts $h$ to extract repeating patterns. This enables us to describe four main building blocks, from which every circuit is composed, schematically shown for $k=4$ in \cref{fig:building_blocks}.
In the following, we derive closed-form descriptions for the three compilation strategies, depending on Hadamard density. To streamline presentation and make subsequent calculations cleaner and more concise, we first introduce a few shorthand notations:
Let $G$ denote an arbitrary single-qubit gate. For an index set
$I \subseteq [n]$, we define the corresponding single-qubit gate layer by
\begin{align}
G_{I} := \prod_{i \in I} G_i .
\end{align}
For two index sets $I_1, I_2 \subseteq [n]$, we define a CZ block acting on all
pairs $(i,j) \in I_1 \times I_2$ as
\begin{align}
\mathrm{CZ}_{(i,j) \in I_1 \times I_2}
:= \prod_{i \in I_1} \ \prod_{j \in I_2} \mathrm{CZ}_{i,j}.
\end{align}
Analogously, for a CNOT block with controls on $I_1$ and targets on $I_2$,
we define
\begin{align}
\mathrm{CX}_{i \in I_1 \to j \in I_2}
:= \prod_{i \in I_1} \ \prod_{\substack{j \in I_2 \\ j \neq i}}
\mathrm{CX}_{i \to j}.
\end{align}

\begin{figure}[tbp]
  \centering
  \includegraphics[width=0.7\linewidth]{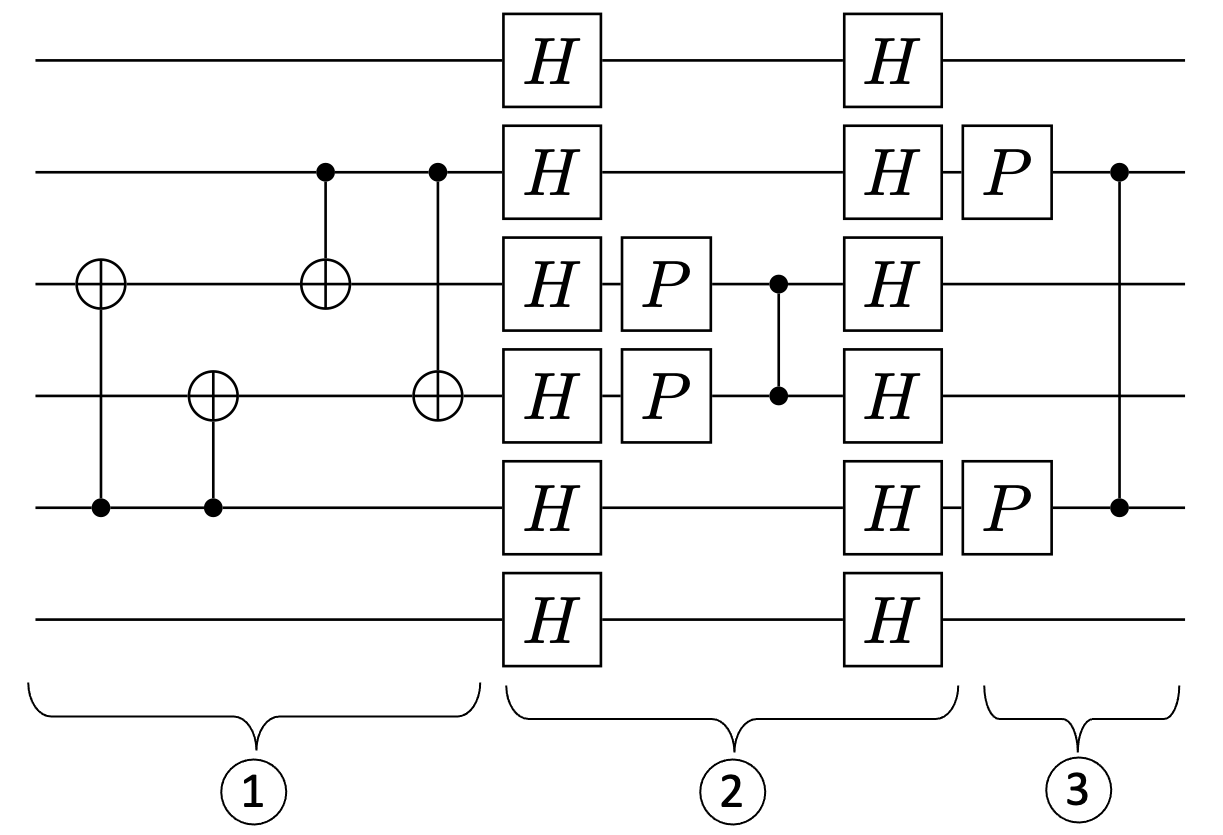}
  \caption{Building blocks for logical realizations of \glspl{cqsk} with even Hadamard count: (1) CNOT-entanglement layer, (2) IQP-like structure, (3) Z-diagonal layer. While shown for $k=4$, the same structure is observed for all system sizes. The concrete instantiation of the blocks depends on the Hadamard count, position, and the choice of strategy $\mathcal{S}$.}
  \label{fig:building_blocks}
\end{figure}

\begin{figure}
    \centering
    \subfigure[\label{fig:block_1}Two instances of the first block from \cref{fig:building_blocks}: the targets of the CNOT operations are wires in $\mathcal{E}(I_h)$, while the controls lie in the disjoint set $\mathcal{E}(\overline{I_h})$. The structure is subdivided into CNOT operations with (a) control wire indices larger than target wire indices and (b) vice versa. Overall, this observation is summarized in \cref{eq:sequencing_rule_1}. Note that $I_h, \overline{I_h} \subseteq \lbrack k \rbrack$.]{
        \includegraphics[width=0.99\linewidth]{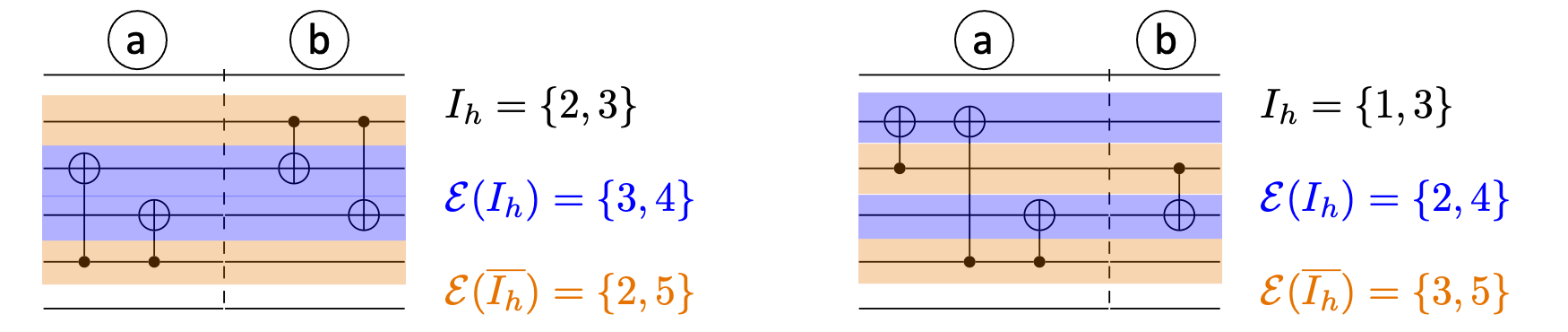}
    }\\
    \subfigure[\label{fig:block_2}Two instances of the second block from \cref{fig:building_blocks}: the first and last layer contain Hadamard gates on all wires. In between, $P$ gates ($CZ$ gates) act on all (pairs of) wires in the index set $\mathcal{E}(I_h)$. With some of the Hadamard gates canceling to identity, we arrive at \cref{eq:sequencing_rule_2}.]{
        \includegraphics[width=0.99\linewidth]{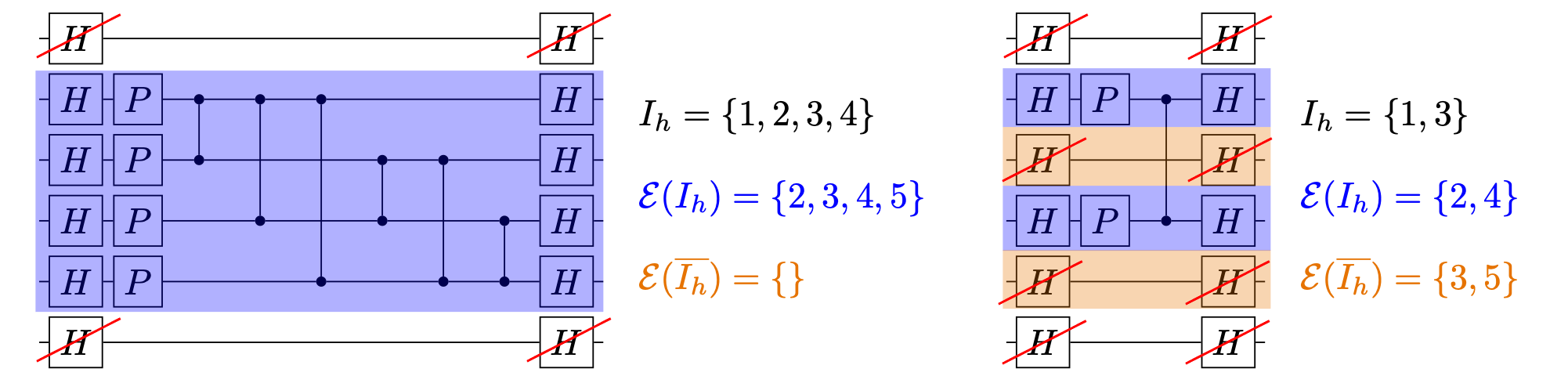}
    }\\
    \subfigure[\label{fig:block_3}Three instances of the third block from \cref{fig:building_blocks}: $P$ gates ($CZ$ gates) act on all (pairs of) wires in the index set $\mathcal{E}(\overline{I_h})$, yielding \cref{eq:sequencing_rule_3}.]{
        \includegraphics[width=0.99\linewidth]{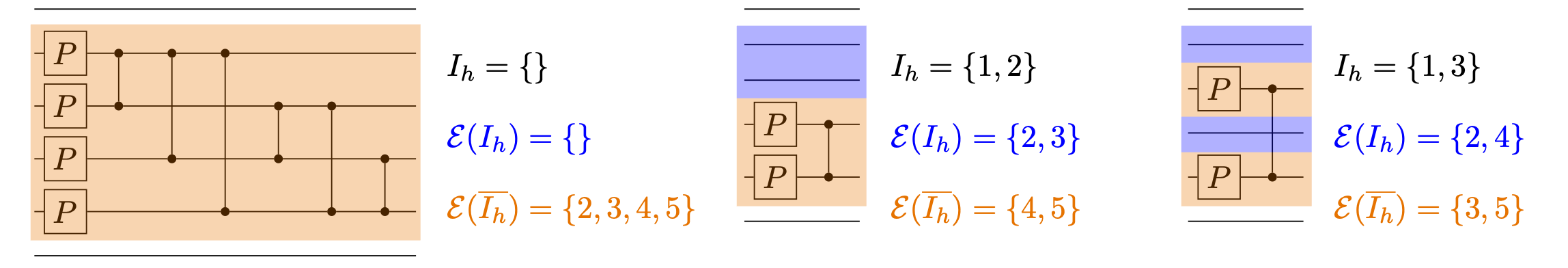}
    }
    \caption{Exemplary instances of the three blocks from \cref{fig:building_blocks}, corresponding to \gls{lcs} solution $\mathcal{S}^{H}_{\text{mid}}$. The underlying \gls{cqsk} circuits are of size $k=4$ and contain an even number of Hadamard gates. 
    Based on just these few instances, it is possible to compactly formulate the compilation strategy in \cref{eq:sequencing_rule_1,eq:sequencing_rule_2,eq:sequencing_rule_3}. The full compilation strategy $\mathcal{C}^{H}_{\text{mid}}$ just sequentially concatenates the three compiled blocks.}
    \label{fig:blocks}
\end{figure}

\subsection{\label{subsec:compilation_strategy_moderate}Compilation Strategy for Moderate Hadamard Counts}
To determine the optimal compilation strategy for moderate Hadamard counts, interpreted as the range defined in \cref{eq:selection_rules}, we analyze the four building blocks identified in \cref{fig:building_blocks} for instances of size $k=4$. However, we note that the same structures generalize to systems of arbitrary size, as we also prove at the end of this section.

In \cref{fig:block_1} we show representative logical realizations of the first block, with the sets $\mathcal{E}(I_h)$ and $\mathcal{E}(\overline{I_h})$ highlighted in blue and red, respectively. It is evident that the CNOT-targets lie on $\mathcal{E}(I_h)$, while the controls act on $\mathcal{E}(\overline{I_h})$. Moreover, the CNOT-structure can be divided into two parts:
First, all CNOTs are applied with control indices larger than their targets, 
while in the second part, the order is reversed, with controls smaller than the targets. 
This decomposition can be expressed as $(a)~\mathrm{CX}_{ \substack{i \in \mathcal{E}(\overline{I_h}) \to j \in \mathcal{E}(I_h) \\ j < i}}$ and $(b)~
\mathrm{CX}_{ \substack{i \in \mathcal{E}(\overline{I_h}) \to j \in \mathcal{E}(I_h) \\ j > i} }$,
which we summarize into
\begin{align}
    \label{eq:sequencing_rule_1}
    \mathrm{CX}_{ \mathcal{E}(\overline{I_h}) \to \mathcal{E}(I_h)}.
\end{align}

Similar regularities can be identified in the second block consisting of an IQP-like structure~\cite{shepherd2009temporally}, as shown in \cref{fig:block_2}. First, many of the Hadamard cancel since $H^2 = \mathbb{I}$. The remaining Hadamard gates always occupy the wires with indices in $\mathcal{E}(I_h)$. The same applies to the consecutive layer of $P$ gates. Finally, both controls and targets of the $CZ$ gates are restricted to unordered combinations from $\mathcal{E}(I_h)$. Altogether, the structure can be summarized by the following sequence:
\begin{align}
\label{eq:sequencing_rule_2}
H_{\mathcal{E}(I_h)} \cdot P_{\mathcal{E}(I_h)} \cdot \mathrm{CZ}_{\{(i,j) \in \mathcal{E}(I_h)^2 \mid\; j>i\}} \cdot H_{\mathcal{E}(I_h)}
\end{align}

We now turn our attention to the third and final structural block of the compiled circuit: a Z-diagonal layer consisting of $P$ gates followed by $CZ$ gates. In contrast to the second block, these gates operate solely on wires with indices in $\mathcal{E}(\overline{I_h})$.  
This behavior can be summarized by the following product rules:
\begin{align}
\label{eq:sequencing_rule_3}
P_{\mathcal{E}(\overline{I_h})} \cdot \mathrm{CZ}_{\{(i,j) \in \mathcal{E}(\overline{I_h})^2 \mid\; j>i\}},
\end{align}
where the first factor represents the application of $P$ gates across all wires in $\mathcal{E}(\overline{I_h})$, while the second one indicates a layer of $CZ$ gates acting pairwise within $\mathcal{E}(\overline{I_h})$.

\medskip
\noindent
By combining all three building blocks, we recover the same gate sequencing rule as obtained by the \gls{sas} method, which is depicted in \cref{fig:sas_rules}. The case for odd Hadamard gates works analogously to the above, for which the same equivalence observation could be made. However, compared to the approach in the original \gls{sas} paper~\cite{chen2024tailoring}, above derivation follows a much more straightforward and intuitive paradigm: instead of going through the intricate procedure of stitching together root circuits derived from Pauli constraints and subsequently applying logical identities to simplify the circuit and reduce its depth, we directly identify simple structural rules for the constituent blocks and compose them.  
\ifshort
\else
    As the methodology for deriving the two other compilation strategies $\mathcal{C}^{H}_{\text{low}}$ and $\mathcal{C}^{H}_{\text{high}}$ is the same as above, we defer them to App.~\ref{app:other_strategies}.
\fi

To prove that the derived constructions are indeed logical realizations, it remains to be shown that the gate-sequencing rules satisfy two key requirements: they implement the desired logical mapping of the \gls{cqsk} and preserve the stabilizers on the $[[n,n-2,2]]$ code. We state respective theorems in the following and defer the straightforward but laborious proofs to 
\ifshort
\else
    App.~\ref{app:proofs}.
\fi

\begin{theorem}
    \label{the:pauli}
    The compilation strategies $\mathcal{C}^{H}_{\text{low}}, \mathcal{C}^{H}_{\text{mid}}, \mathcal{C}^{H}_{\text{high}}$ satisfy the physical Pauli constrains for any \gls{cqsk} circuit.
\end{theorem}

\begin{theorem}
    \label{the:stabilizer}
    The compilation strategies $\mathcal{C}^{H}_{\text{low}}, \mathcal{C}^{H}_{\text{mid}}, \mathcal{C}^{H}_{\text{high}}$ preserve the $[[n,n\mathrm{-}2,2]]$ code stabilizers $X_{[n]}$ and $Z_{[n]}$.
\end{theorem}

With that, we derived provably correct and stabilizer-preserving compilation strategies $\mathcal{C}^{H}_{\text{low}}$, $\mathcal{C}^{H}_{\text{mid}}$, $\mathcal{C}^{H}_{\text{high}}$ which are depth-optimized for varying Hadamard counts. As the $\mathcal{C}^{H}_{\text{mid}}$ compilation strategy is equivalent to the \gls{sas} approach in \cref{fig:sas_rules}, for the mid-regime in \cref{eq:selection_rules}, both methods produce compiled circuits of the same depth. However, as shown in \cref{sec:compilation_primitives} for sparse and dense Hadamard placement, one can significantly improve depth by switching to other strategies $\mathcal{C}^{H}_{\text{low}}, \mathcal{C}^{H}_{\text{high}}$. With that, we emphasize that the combination of the three strategies allows for the realization of the \gls{pbs} compiler introduced in \cref{sec:compilation_primitives}, which we will empirically evaluate in the following.


\section{\label{sec:evaluation}Empirical Evaluation}

In this section, we empirically evaluate the \gls{pbs} compilation strategy with the selection terms given in \cref{eq:selection_rules} switching between $\mathcal{C}^{H}_{\text{low}}$, $\mathcal{C}^{H}_{\text{mid}}$, and $\mathcal{C}^{H}_{\text{high}}$, respectively. For the experiments, we implemented compilation, noise, and evaluation of \glspl{cqsk} on the $[[n,n\mathrm{-}2,2]]$ code using Qiskit~\cite{Qiskit}. The simulations are conducted using the $\texttt{AerSimulator}$ with the \texttt{extended-stabilizer} method.

The evaluation procedure is given a system size $k$, a Hadamard-position set $I_h$, and depolarizing noise rates $p_1$ and $p_2$ (non-correlated for single-qubit gates and correlated for two-qubit gates, respectively). We then generate a logical \gls{cqsk} circuit both via the \gls{sas} approach and using our \gls{pbs} compilation strategy. The assembled pipeline comprises \gls{qec} \emph{encoder}, the respective \emph{logical \gls{cqsk}} realization, \emph{syndrome detection}, and \emph{measurement} of both data wires and stabilizers. To isolate the effect of noise on the logical computation, the encoder is compiled into a single Clifford instruction and treated as noiseless, and the syndrome extraction and measurement stage is likewise assumed ideal.

The syndrome-detection measurement, together with the data readout yields empirical distributions over outcomes of the form \((x,s)\), where \(x \in \{0,1\}^n\) is the \(n\)-bit string obtained from the data register, and \(s \in \{00,01,10,11\}\) is the two-bit syndrome string obtained from the stabilizer measurement. We define the \emph{acceptance rate} $p_{\mathrm{acc}}$ as the fraction of outcomes, for which the syndrome yields $00$. Complimentarily, the \emph{success rate} $p_{\mathrm{succ}}$ is the fraction of runs, where the $00$ syndrome is observed, and there was actually no fault in the computation. 

\begin{figure}[tbp]
  \centering
  \includegraphics[width=0.8\linewidth]{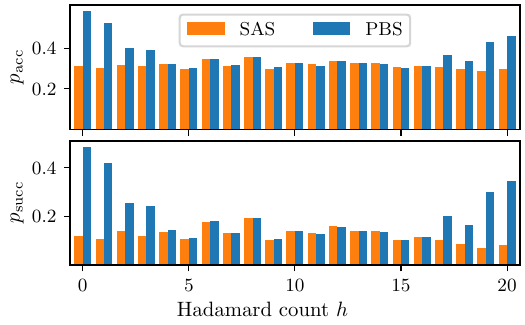}
  \caption{Comparison of \gls{sas} and \gls{pbs} compilation strategies for \glspl{cqsk} on $k=20$ qubits across Hadamard counts $h$. We analyze the (upper plot) acceptance rate $p_{\mathrm{acc}}$ and the (lower plot) success rate $p_{\mathrm{succ}}$, identifying advantages of \gls{pbs} in the edge regimes of sparse and dense Hadamard placements.}
  \label{fig:evaluation}
\end{figure}

For a system size of $k=20$ and noise strength $p_1=0.01$ and $p_2=0.01$ (evaluations with more realistic noise models are deferred to future work), we compare $p_{\mathrm{acc}}$ and $p_{\mathrm{succ}}$ resulting from \gls{sas} and \gls{pbs} compilation in \cref{fig:evaluation}. As suggested by the equivalence of \gls{sas} and the compilation strategy $\mathcal{C}^{H}_{mid}$ selected by \gls{pbs} for moderate Hadamard counts (see \cref{eq:selection_rules}), for $6 \leq h \leq 14$ both methods are on par up to minor statistical fluctuations. However, in the edge regimes, i.e.\ low Hadamard counts $h \leq 5$ and high counts $h \geq 15$, the strategies $\mathcal{C}^{H}_{low}$ and $\mathcal{C}^{H}_{high}$ are selected by \gls{pbs}, respectively. For these circuit structures, one can observe a significantly higher $p_{\mathrm{acc}}$ and $p_{\mathrm{succ}}$, likely caused by the shallower logical realization and two-qubit count being the dominant error driver in realistic mid-term setups. Empirically, we observed that for larger $k$ the differences between \gls{sas} and \gls{pbs} in the edge regimes become even more pronounced, highlighting the promise of the newly introduced compilation strategy.


\glsresetall
\section{\label{sec:discussion}Discussion and Outlook}

In this work, we introduced a novel technique for the logical compilation of Clifford circuits on particular code families. Starting with small \gls{cqsk} instances, we enumerated the full \gls{lcs} solution space~\cite{rengaswamy2018synthesis}, profiled depth and two-qubit counts across circuit configurations, and extracted depth-favorable compilation primitives. We then distilled these components into compilation strategies that generalize to arbitrary system sizes, yielding the 
\gls{pbs} routine that lowers two-qubit gate counts relative to \gls{sas}~\cite{chen2024tailoring,chen2025tailoring}, which itself already significantly improved upon gate-by-gate compilation. The gains are particularly prevalent for \gls{cqsk} with sparse and dense Hadamard placements. In our noise studies, we identified circuit depth as the dominant driver of logical performance; correspondingly, due to reduced two-qubit gate count, we also observed improvements in success probability.

Despite these promising insights, limitations remain: the method is demonstrated on the $[[n,n\mathrm{-}2,2]]$ code family (supporting error detection but not correction), focuses on \gls{cqsk} and thus Clifford circuits, and does not yet enforce full fault tolerance. We view the first two as matters of scope rather than principle: because our workflow mines small instances and generalizes primitives into closed-form templates verified by Pauli-constraint and stabilizer-preservation checks, it can be applied to other code families with higher distances. The same flexibility allows broadening beyond \gls{cqsk} to other Clifford subroutines, enhancing usefulness in a peephole optimization workflow. While the approach does not generalize to arbitrary non-Clifford routines, it may extend to circuits with limited non-Clifford components, e.g., a single $T$ gate within the quantum simulation kernel. However, such considerations are beyond the scope of this work. The third limitation, i.e..\ lack of enforced fault tolerance can be addressed by embedding flag gadgets~\cite{chao2018fault}, as demonstrated in~\cite{chen2024tailoring}.

In summary, our approach provides a compact and extensible tool for improving peephole-based logical compilation: it identifies compilation primitives on small instances and generalizes these with limited hand-crafting overhead to scalable compilation strategies.

\section*{Acknowledgment}
We thank N. Rengaswamy and Z. Chen for insightful discussions of their work, which inspired this study.
We furthermore thank J. Jordan for administrative support and constructive feedback.

\smallskip
\noindent The authors used GPT-5.1 for language editing of the paper. All content was reviewed and edited by the authors, who take full responsibility for the final work.

\ifshort
\else
    \section*{Code Availability}
    The code supporting this study is currently undergoing final polishing and will be released publicly at a later stage. Further information and data is available upon reasonable request.
\fi

\ifshort
\else
    \glsresetall
    \appendix
    \subsection{\label{app:other_strategies}Other Compilation Strategies for Piecewise-Best Selection}

In the following, we derive the compilation strategies $\mathcal{C}^{H}_{low}$ for low Hadamard count in App.~\ref{subapp:low_hadamard} and $\mathcal{C}^{H}_{high}$ for high Hadamard count in App.~\ref{subapp:high_hadamard}. In large parts, the methodology is the same as in \cref{subsec:compilation_strategy_moderate}, but a few additional insights are introduced. As already explained in the main part, the combination of the three compilation strategies allows us to define the \gls{pbs} strategy, which produces depth-optimized logical realizations in all Hadamard regimes.


\medskip
\subsubsection{\label{subapp:low_hadamard}Compilation Strategy for Low Hadamard Counts}
While the derivations only slightly differ, there are different structural patterns for \glspl{cqsk} with even and with odd Hadamard count, which we consider separately in the following:

\smallskip
\noindent\emph{Even number of Hadamard gates.} Comparing circuits generated with $\mathcal{S}^{H}_{low}$ to $\mathcal{S}^{H}_{mid}$, we observe that the first two blocks (i.e.\ the CNOT-entanglement layer and the IQP-like structure) are identical and thus follow \cref{eq:sequencing_rule_1,eq:sequencing_rule_2}. The only difference lies in the final Z-diagonal block, where the additional $P$-layer follows the rule
\begin{align}
    \label{eq:low_block_3_1}
    P_{[n]\setminus \mathcal{E}(\overline{I_h})},
\end{align}
that is, $P$ gates are applied on all wires except those belonging to $\mathcal{E}(\overline{I_h})$.
For the CZ-structure, instead of directly analyzing which sequencing rules generate the circuit blocks, it is more instructive to examine what is missing compared to the fully connected block. For the instance shown in \cref{fig:low_even_CZ_block}, exactly those $CZ$ gates whose control and target both lie in $\mathcal{E}(\overline{I_h})$ are absent. The missing gate is highlighted by the red oval. This pattern holds systematically for all \glspl{cqsk}, leading to the following compact formula:
\begin{align}
    \label{eq:low_block_3_2}
    \mathrm{CZ}_{\{(i,j) \in [n]^2\setminus\mathcal{E}(\overline{I}_h)^2 \mid\; j>i \}}
\end{align}
Combining \cref{eq:low_block_3_1,eq:low_block_3_2} gives us the sequencing rule for the third block of $\mathcal{C}^{H}_{low}$ in the case of even Hadamard count.

\smallskip
\noindent\emph{Odd number of Hadamard gates.} As before, the first and second block in this case are identical to the strategy for $\mathcal{S}^{H}_{mid}$. Analyzing representative examples in \cref{fig:low_odd_third_block}, we get the  sequencing formula for the $P$ gates as
\begin{align}
    \label{eq:low_odd_block_3_1}
    P_{i \in \{1\} \cup \mathcal{E}(I_h)}.
\end{align}
Adopting the convention that for each $CZ$ gate the control is the upper wire and the target on the lower wire, the following pattern holds: for every wire $i \in [n] \setminus \mathcal{E}(\overline{I_h})$ we obtain  
\begin{align}
    \mathrm{CZ}_{\{(i,j) \in \{i\} \times [n] \mid\; j>i\}},
\end{align}
while for every wire $i \in \mathcal{E}(\overline{I_h})$ we obtain the block  
\begin{align}
    \mathrm{CZ}_{\{(i,j) \in \{i\} \times \mathcal{E}(I_h) \mid\; j>i\}}.
\end{align}
In \cref{fig:low_odd_third_block}, these two types of blocks are separated by dashed lines. By iterating through all wires, the complete sequencing rule for this block can be written as  
\begin{align}
    \label{eq:low_odd_block_3_2}
    \mathrm{CZ}_{\{(i,j) \in [n]\setminus \mathcal{E}(\overline{I_h}) \times [n] \mid \; j > i\}} \cdot \mathrm{CZ}_{\{(i,j) \in \mathcal{E}(\overline{I_h}) \times \mathcal{E}(I_h) \mid\; j > i\}},
\end{align}
which combined with \cref{eq:low_odd_block_3_1} completes the set of sequencing rules for $\mathcal{C}^{H}_{low}$.

\begin{figure}[tbp]
  \centering
  \subfigure[\label{fig:low_even_CZ_block}Part of third block of \cref{fig:building_blocks} for even Hadamard count. Compared to the fully connected instance, the $CZ$ gate on the wire pair $\mathcal{E}(\overline{I_h})^{2}=\{(3,5)\}$ is omitted, yielding \cref{eq:low_block_3_2}.
]{
    \includegraphics[width=0.99\linewidth]{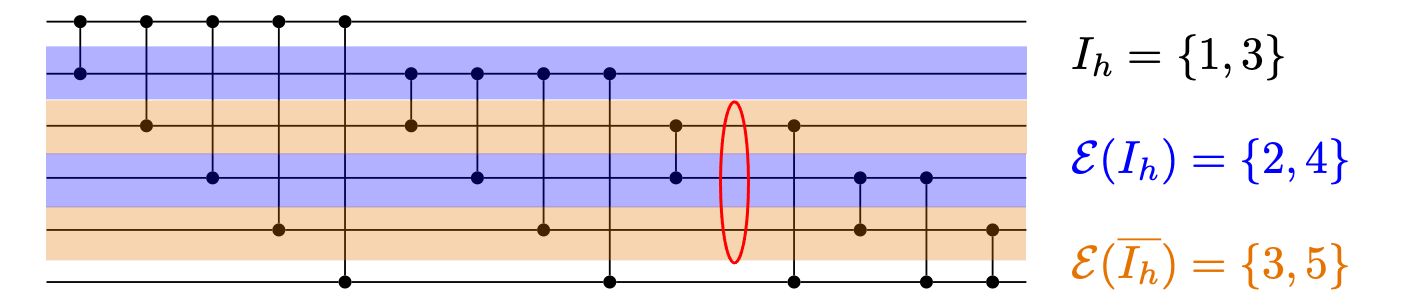}
  }\\
  \subfigure[\label{fig:low_odd_third_block}Third block of \cref{fig:building_blocks} for odd Hadamard count. The \(P\)-layer acts on the first wire and wires in $\mathcal{E}(I_h)$. Wire by wire, if the upper control lies in $\lbrack n\rbrack\setminus \mathcal{E}(I_h)$, all lower qubits are targeted by $CZ$ gates; otherwise, an upper qubit in $\mathcal{E}(I_h)$ couples via $CZ$ only to qubits in $\mathcal{E}(I_h)$, yielding \cref{eq:low_odd_block_3_2}.
]{
    \includegraphics[width=0.99\linewidth]{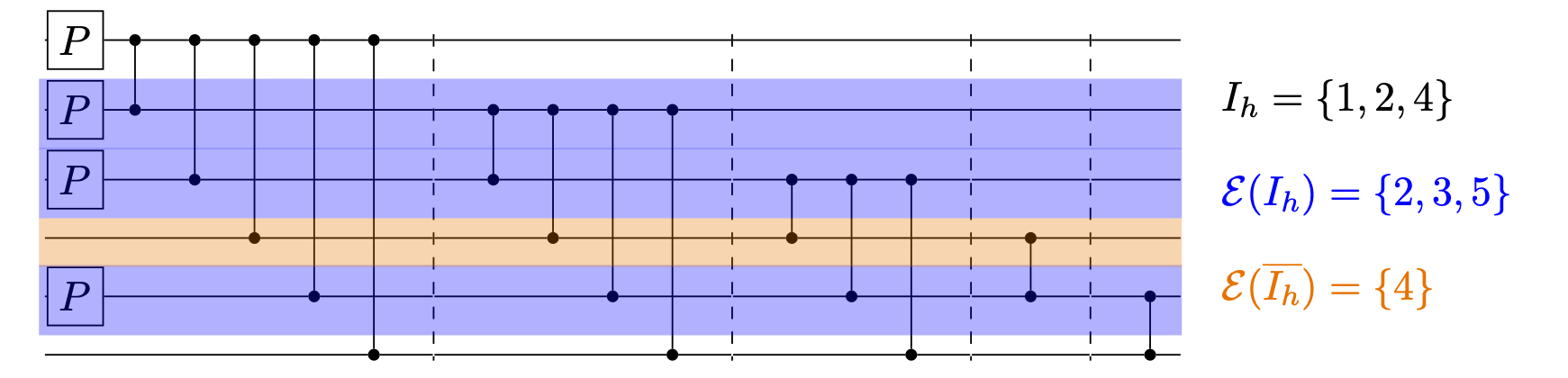}
  }
  \caption{Exemplary instances of the blocks from \cref{fig:building_blocks}, corresponding to \gls{lcs} solution $\mathcal{S}^{H}_{\text{low}}$. The underlying \gls{cqsk} circuits are of size $k=4$.}
\end{figure}


\medskip
\subsubsection{\label{subapp:high_hadamard}Compilation Strategy for High Hadamard Counts}
As the derivations for $\mathcal{C}^{H}_{high}$ follow the same ideas as for $\mathcal{C}^{H}_{low}$ and $\mathcal{C}^{H}_{mid}$, we directly state the sequencing rules in the following.

\smallskip
\noindent\emph{Even number of Hadamard gates.}
The sequencing rules for the first and third block are identical between $\mathcal{C}^{H}_{mid}$ and $\mathcal{C}^{H}_{high}$. For the second layer, it is easy to see that Hadamard gates are applied on every wire $i \in [n]$ and that the $P$ gates act only on wires in $[n] \setminus \mathcal{E}(I_h)$, which yields
\begin{align}
H_{[n]} \qquad \text{and} \qquad P_{[n] \setminus \mathcal{E}(I_h)}.
\end{align}
Proceeding as above and comparing the CZ structure with the fully connected circuit, we obtain
\begin{align}
\mathrm{CZ}_{\{(i,j) \in [n]^2\setminus \mathcal{E}(I_h)^2 \mid \; j>i \}}.
\end{align}

\smallskip
\noindent\emph{Odd number of Hadamard gates.}
This last case just replaces the sequencing rules for the second layer with
\begin{align}
H_{[n]} \qquad \text{and} \qquad
P_{\mathcal{E}(\overline{I_h}) \cup \{n\}},
\end{align}
as well as  
\begin{align}
\mathrm{CZ}_{(i,j) \in [n]^2\setminus\big( (\mathcal{E}(I_h) \cup \{1\})\times \mathcal{E}(I_h) \big)}.
\end{align}


\subsection{\label{app:proofs}Proof of Theorems}

In the following, we prove \cref{the:pauli,the:stabilizer} showing the correctness of the compilation strategies we derived.

\smallskip
\subsubsection{Physical Pauli Constraints}
\begin{theorem*}
    The compilation strategies $\mathcal{C}^{H}_{\text{low}}, \mathcal{C}^{H}_{\text{mid}}, \mathcal{C}^{H}_{\text{high}}$ satisfy the physical Pauli constraints for any \gls{cqsk} circuit.
\end{theorem*}
\noindent
We explicitly execute the proof for $\mathcal{C}^{H}_{mid}$ in the case of an even number of Hadamard gates and for cases where $k \in \mathcal{E}(\overline{I_h})$ is the mapping relation for the physical operators $X_1 X_k$, which is given by 
\begin{align}
X_1X_{k} \mapsto
\begin{cases}
\begin{aligned}
    &X_1X_{k}, && k\in \mathcal{E}(I_h),\\
    &X_1Y_{k}X_{\mathcal{E}(I_h)}Z_{\mathcal{E}(\overline{I_h})\setminus \{k\}},
      && k\in \mathcal{E}(\overline{I_h}).
\end{aligned}
\end{cases}
\end{align}
All remaining mapping relations, namely the case $k \in \mathcal{E}(I_h)$ for $X_1 X_k$,
the relations for $Z_k Z_n$ for both $k \in \mathcal{E}(I_h)$ and
$k \in \mathcal{E}(\overline{I_h})$, the corresponding cases for an odd number
of Hadamard gates, as well as the other compilation strategies $\mathcal{C}^{H}_{low}$, $\mathcal{C}^{H}_{high}$, can be shown in complete analogy by following the same procedure.

\begin{proof}
The first gate block of the operator $X_1 X_k$ is conjugated by $\mathrm{CX}_{\mathcal{E}(I_h)\to \mathcal{E}(\overline{I_h})}$:
\begin{align}
&\mathrm{CX}_{\mathcal{E}(I_h) \to \mathcal{E}(\overline{I_h})}
\left( X_1 X_k \right)
\mathrm{CX}_{\mathcal{E}(I_h) \to \mathcal{E}(\overline{I_h})}
= X_1 X_k X_{\mathcal{E}(\overline{I_h})}
\end{align}
The conjugation with the $H_{\mathcal{E}(I_h)}$ and $P_{\mathcal{E}(I_h)}$ layer only effects qubits on the $\mathcal{E}(I_h)$ wires, hence 
\begin{align}
&P_{\mathcal{E}(I_h)}H_{\mathcal{E}(I_h)}
\left( X_1 X_k X_{\mathcal{E}(I_h)} \right)
H_{\mathcal{E}(I_h)}P_{\mathcal{E}(I_h)}
= X_1 X_k Z_{\mathcal{E}(I_h)} .
\end{align}
The conjugation by $\mathrm{CZ}_{\{(i,j)\in \mathcal{E}(I_h)^2 \mid j> i\}}$ has no effect on $X_1 X_k Z_{\mathcal{E}(I_h)}$, since qubits $1$ and $k$ lie outside $\mathcal{E}(I_h)$ and $Z$ gates commute with $CZ$ gates.
The subsequent $H_{\mathcal{E}(I_h)}$ transforms $Z_{\mathcal{E}(I_h)}$ into $X_{\mathcal{E}(I_h)}$, acting only on qubits in $\mathcal{E}(I_h)$.
Next, $P_{\mathcal{E}(\overline{I_h})}$ acts only on $X_k$ (because $1 \notin \mathcal{E}(\overline{I_h})$ and $\mathcal{E}(I_h) \cap \mathcal{E}(\overline{I_h}) = \varnothing$) and turns it into $Y_k$.
Hence, before the final $\mathrm{CZ}$-block, the operator becomes $X_1 Y_k X_{\mathcal{E}(I_h)}$. Conjugation by $\mathrm{CZ}_{\{(i,j)\in \mathcal{E}(\overline{I_h})^2 \mid j> i\}}$ then yields

\begin{align}
&\mathrm{CZ}_{\{(i,j)\in \mathcal{E}(\overline{I_h})^2 \mid j> i\}}
\left( X_1 Y_k
X_{\mathcal{E}(I_h)} \right)
\mathrm{CZ}_{\{(i,j)\in \mathcal{E}(\overline{I_h})^2 \mid j> i\}}\notag
\\[0.3em]
&= X_1 Y_k Z_{\mathcal{E}(\overline{I_h})\setminus \{k\}} X_{\mathcal{E}(I_h)}.
\end{align}
\end{proof}


\smallskip
\subsubsection{Stabilizer Preservation}

\begin{theorem*}
    The compilation strategies $\mathcal{C}^{H}_{\text{low}}, \mathcal{C}^{H}_{\text{mid}}, \mathcal{C}^{H}_{\text{high}}$ preserve the $[[n,n\mathrm{-}2,2]]$ code stabilizers $X_{[n]}$ and $Z_{[n]}$.
\end{theorem*}
\noindent
We show the proof for $X$-type stabilizer on $\mathcal{C}^{H}_{mid}$, while the proofs for the $Z$-type stabilizer and respective versions for $\mathcal{C}^{H}_{low}$, $\mathcal{C}^{H}_{high}$ can be done in complete analogy.

\begin{proof}
To prove the preservation of the stabilizer
$X_{[n]} = X_1\, X_{\mathcal{E}(I_h)}\, X_{\mathcal{E}(\overline{I_h})}\, X_n$, we have to perform the same procedure on this operator as before.
Conjugating $X_{[n]}$ by $\mathrm{CX}_{\mathcal{E}(\overline{I_h})\to \mathcal{E}(h)}$ we only have
to analyze its working on $X_{\mathcal{E}(I_h)}\,X_{\mathcal{E}(\overline{I_h})}$ since
$1,n \notin \mathcal{E}(I_h), \mathcal{E}(\overline{I_h})$. Here, we have
\begin{align}
&\mathrm{CX}_{\mathcal{E}(\overline{I_h})\to\mathcal{E}(I_h)}\,
X_{\mathcal{E}(I_h)}X_{\mathcal{E}(\overline{I_h})}\,
\mathrm{CX}_{\mathcal{E}(\overline{I_h})\to\mathcal{E}(I_h)}\notag
\\[0.3em]
&=\bigl(\mathrm{CX}_{\mathcal{E}(\overline{I_h})\to\mathcal{E}(I_h)}X_{\mathcal{E}(I_h)}\mathrm{CX}_{\mathcal{E}(\overline{I_h})\to\mathcal{E}(I_h)}\bigr)\notag
\\[0.3em]
&\cdot \bigl(\mathrm{CX}_{\mathcal{E}(\overline{I_h})\to\mathcal{E}(I_h)}X_{\mathcal{E}(\overline{I_h})}\mathrm{CX}_{\mathcal{E}(\overline{I_h})\to\mathcal{E}(I_h)}\bigr).
\end{align}
It is $\mathrm{CX}_{\mathcal{E}(\overline{I_h})\to\mathcal{E}(I_h)}\,X_{\mathcal{E}(I_h)}\,\mathrm{CX}_{\mathcal{E}(\overline{I_h})\to\mathcal{E}(I_h)}
= X_{\mathcal{E}(I_h)}$ because 
\\[0.3em]
$\mathrm{CX}_{i\to j}\,X_j\,\mathrm{CX}_{i\to j}=X_j$ for all $i \in \mathcal{E}(\overline{I_h})$ and $j \in   \mathcal{E}(I_h)$. Further, it holds that 
\begin{align}
&\mathrm{CX}_{\mathcal{E}(\overline{I_h})\to\mathcal{E}(I_h)}\,
X_{\mathcal{E}(\overline{I_h})}\,
\mathrm{CX}_{\mathcal{E}(\overline{I_h})\to\mathcal{E}(I_h)}\notag
\\[0.3em]
&= \prod_{i\in \mathcal{E}(\overline{I_h})} X_i\,X_{\mathcal{E}(I_h)}
= X_{\mathcal{E}(\overline{I_h})}\,
\underbrace{X_{\mathcal{E}(I_h)}^{\,|\mathcal{E}(\overline{I_h})|}}_{=\mathbb{I}}
= X_{\mathcal{E}(\overline{I_h})}
\end{align}
showing that the CNOT block leaves the X stabilizer invariant. 
Application of $H_{\mathcal{E}(I_h)}$ and than $P_{\mathcal{E}(I_h)}$ effects only
the $X_{\mathcal{E}(I_h)}$ gates in the stabilizers turning it into $Z_{\mathcal{E}(I_h)}$.
Further application of the $\mathrm{CZ}_{\{(i,j)\in \mathcal{E}(I_h)^2 \mid j>i\}}$
block on the conjugated $X$-stabilizers has no effect since
$1,n \notin \mathcal{E}(I_h)$ and $\mathcal{E}(I_h),\mathcal{E}(\overline{I_h})$
are disjoint sets.
Conjugation again with $H_{\mathcal{E}(I_h)}$ transforms $Z_{\mathcal{E}(I_h)}$
back to $X_{\mathcal{E}(I_h)}$ and $P_{\mathcal{E}(\overline{I_h})}$ turns
$X_{\mathcal{E}(\overline{I_h})}$ into $Y_{\mathcal{E}(\overline{I_h})}$, hence in
total we have $X_1\,X_{\mathcal{E}(I_h)}\,Y_{\mathcal{E}(\overline{I_h})}\,X_n$. The last conjugation with the $\mathrm{CZ}_{\{(i,j)\in \mathcal{E}(\overline{I_h})^2 \mid j> i\}}$ block acts only on $Y_{\mathcal{E}(\overline{I_h})}$. First, we rewrite the CZ-block in the following way:
\begin{align}\label{eq:PM}
\mathrm{CZ}_{\{(i,j)\in \mathcal{E}(\overline{I_h})^2 \mid j> i\}} 
=\prod_{\substack{M \in \mathcal{M}(\mathcal{E}(\overline{I_h}))\\ (i,j) \in M}} \mathrm{CZ}_{i,j}
\end{align}
Here, $\mathcal{M}(\mathcal{E}(\overline{I_h}))$ is the set of all perfect matchings for the graph $G(V,E)$ with $V = \mathcal{E}(\overline{I_h})$ and $E = \{(i,j) \mid i<j\}$. The graph is complete, hence the number of perfect matchings is odd. For an $M \in \mathcal{M}(\mathcal{E}(\overline{I_h}))$, it is 
\begin{align}
\mathrm{CZ}_M\, Y_{\mathcal{E}(\overline{I_h})}\, \mathrm{CZ}_M
= (-1)^{|M|} X_{\mathcal{E}(\overline{I_h})},\\
\mathrm{CZ}_M\, X_{\mathcal{E}(\overline{I_h})}\, \mathrm{CZ}_M
= (-1)^{|M|} Y_{\mathcal{E}(\overline{I_h})}. 
\end{align}
Furthermore, we use the fact that for two indices $i,j \in \mathcal{E}(\overline{I_h})$, it holds that $\mathrm{CZ}_{i,j}Y_iY_j\mathrm{CZ}_{i,j} = -X_iX_j$ and $\mathrm{CZ}_{i,j}X_iX_j\mathrm{CZ}_{i,j} = -Y_iY_j$. So applying the CZ-block from \cref{eq:PM} on $Y_{\mathcal{E}(I_h)}$ corresponds to a repeated application of $\mathrm{CZ}_{M}$ blocks with $M \in \mathcal{M}(\mathcal{E}(I_h))$ an odd number of times, which results in $(\pm 1)X_{\mathcal{E}(I_h)}$, hence one can conclude the preservation of the X-stabilizer. 
\end{proof}

\fi

\bibliographystyle{IEEEtran}
\bibliography{paper}

\end{document}